\documentclass[table,svgnames,dvipsnames]{article}

\usepackage{soul}
\usepackage{url}
\usepackage[hidelinks]{hyperref} 
\usepackage[small]{caption}
\usepackage{graphicx}
\usepackage{booktabs}
\usepackage{amsmath,multirow}
\usepackage{amsthm,thm-restate,thmtools}
\usepackage{amssymb}
\usepackage{enumerate}
\usepackage{subcaption}
\usepackage{xcolor,xfrac}
\usepackage[sort,numbers]{natbib} 
\usepackage[english]{babel}
\usepackage{charter}
\usepackage{authblk}
\usepackage{framed}
\usepackage[normalem]{ulem}  
\usepackage{tikz}
\usetikzlibrary{shapes.arrows}
\usetikzlibrary{shapes.geometric} 
\usepackage{amsfonts} 
\usepackage[T1]{fontenc}
\usepackage[utf8]{inputenc}
\usepackage[top=1 in,bottom=1in, left=1 in, right=1 in]{geometry}  
\usepackage{enumitem}
\usepackage{tabularx}
\usepackage{algorithm}
\hypersetup{
     colorlinks   = true,
     linkcolor    = red, % color of internal links
     urlcolor     = teal, % color of external links
	 citecolor    = blue % color of links to bibliography
}
\usepackage{mathtools}
\usepackage{nicefrac} 
\usepackage{pifont}

\usepackage{algorithm}
\usepackage{graphicx}
\usepackage[noend]{algpseudocode}
\usepackage{subcaption}
\usepackage{gensymb}
\usepackage{pgfplots}
\usepackage{bbm}
\usepackage{comment}
\usepackage{xr}

\usepackage{tkz-graph}

\DeclareMathOperator*{\argmax}{arg\,max}
\DeclareMathOperator*{\argmin}{arg\,min}

\DeclareMathOperator*{\opt}{Opt}
\DeclareMathOperator*{\OPT}{OPT}

\DeclareMathOperator*{\roe}{RoE}
\DeclareMathOperator*{\TC}{tc}

\newtheorem{theorem}{Theorem}
\newtheorem{lemma}{Lemma}
\newtheorem{claim}{Claim}

\newtheorem{observation}{Observation}
\newtheorem{remark}{Remark}

\title{Delayed Assignments in Online Non-Centroid Clustering \\ with Stochastic Arrivals}

\date{\empty}

\author{Saar Cohen\\
Department of Computer Science, Bar Ilan University, Ramat Gan, Israel\\
Department of Computer Science, University of Oxford, Oxford, United Kingdom\\
saar30@gmail.com}

\begin{document}

\maketitle

\begin{abstract}
    Clustering is a fundamental problem, aiming to partition a set of elements, like agents or data points, into clusters such that elements in the same cluster are closer to each other than to those in other clusters. In this paper, we present a new framework for studying online non-centroid clustering \textit{with delays}, where elements, that arrive one at a time as points in a finite metric space, should be assigned to clusters, but assignments need \textit{not} be immediate. Specifically, upon arrival, each point's location is revealed, and an online algorithm has to \textit{irrevocably} assign it to an existing cluster or create a new one containing, at this moment, only this point. However, we allow decisions to be postponed at a \textit{delay cost}, instead of following the more common assumption of \textit{immediate} decisions upon arrival. This poses a critical challenge: the goal is to minimize both the total distance costs between points in each cluster and the overall delay costs incurred by postponing assignments. In the classic \textit{worst-case arrival model}, where points arrive in an \textit{arbitrary} order, no algorithm has a competitive ratio better than sublogarithmic in the number of points. To overcome this strong impossibility, we focus on a \textit{stochastic arrival model}, where points' locations are drawn independently across time from an \textit{unknown} and \textit{fixed} probability distribution over the finite metric space. We offer hope for beyond worst-case adversaries: we devise an algorithm that is \textit{\textbf{constant}} competitive in the sense that, as the number of points grows, the ratio between the expected overall costs of the output clustering and an optimal offline clustering is bounded by a constant. 
\end{abstract}

\section{Introduction}

In multiplayer online gaming platforms, players enter the platform over time and are assigned to teams for cooperative gameplay (e.g., completing quests or participating in tournaments). Players usually prefer teammates with similar skill levels, playstyles, and compatible roles. Once a player enters the platform, the platform's team assignment system may postpone the assignment in hope that more compatible players will later join the platform. However, players may become highly unsatisfied if they wait for too long. To optimize their gaming experience, the system should assign players to teams so that players in the same team are more similar to each other than to players in other groups, while minimizing their waiting times for assignments. 

% In ride-sharing platforms, passengers request rides over time and are grouped into shared trips. Passengers typically prefer to share rides with others whose destinations and schedules align. Upon a ride request, the platform may delay assigning a passenger in anticipation of more compatible passengers arriving soon. However, excessive waiting can lead to passenger dissatisfaction. To enhance user satisfaction, the platform aims to form ride groups whose members are more compatible with each other than with those in other groups, while also minimizing waiting times before pickup. Such scenarios and many other real-life cases exemplify the fundamental task of \textit{clustering}, whose goal is to partition a set of elements, such as agents or data points, into \textit{clusters} so that elements within each cluster have similar properties. 

A widely studied model of clustering is \textit{center-based} metric clustering \cite{absalom2022comprehensive}, where elements are points in a metric space whose similarity is measured by distance: closer elements are considered more similar. Clusters are defined by centers, with each point assigned to its nearest center. In contrast, we herein focus on \textit{non-centroid} metric clustering (see, e.g., \cite{Caragiannis2024Proportional}), with no cluster centers. This captures real-world scenarios where elements (e.g., agents) prefer to be close to others in their cluster, like our gaming platform example, clustered federated learning \cite{Sattler2021Clustered}, clustering in social networks \cite{Lee2019Social}, and document clustering \cite{JANANI2019Expert}. Other examples include package-delivery services and ride-sharing platforms. For instance, in last-mile delivery, parcels arrive over time and must be dynamically grouped into vehicle loads, where grouping items with nearby destinations reduces travel distance and improves efficiency. Similarly, in ride-sharing platforms, passengers arrive sequentially and the system must form small groups of riders whose pickup and drop-off locations are mutually compatible. Delaying assignment improves batching opportunities but increases waiting time, exactly the tradeoff modeled in our framework. 

Prior work on non-centroid clustering assumes an \textit{offline} setting, where all elements are known in advance. However, in many real-world scenarios such as our gaming platform example, the input is revealed gradually over time. For that reason, various online clustering problems have received increased attention in recent years (see, e.g., \cite{cohen2022online,lattanzi2021robust,guo2020facility,guo2021Consistent}), yet online \textit{non-centroid metric} clustering has remained largely overlooked. Further, prior works on online clustering often assume that elements are assigned to clusters \textit{immediately} upon arrival, which may lead to poor solutions in practice. For instance, in our gaming platform example, a perfectly compatible teammate might join the platform just after a player has already been assigned to a team.

In this paper, we thus introduce and study a new model of \textit{online} non-centroid metric clustering \textit{with delays}, where elements that arrive one at a time as points in a finite metric space must be assigned to clusters, but assignments need \textit{not} be immediate. Upon arrival, each point's location is revealed. Then, a central authority (i.e., an \textit{online algorithm}) has to \textit{irrevocably} decide whether to assign it to an existing cluster or create a new one containing, at this moment, only this point. However, we allow decisions to be postponed at a \textit{delay cost}, departing from the more common assumption of \textit{immediate} decisions upon arrival. This raises a key challenge: our goal is to minimize both the total distance costs between points in each cluster and the total delay costs incurred by postponing decisions.

If clusters can be of any size, then the non-centroid clustering of singletons is trivially optimal, as it yields zero total cost. The same applies even if we instead seek a \textit{capacitated} clustering, where each cluster must respect an upper bound on its size (see, e.g., \cite{Dinler2016Survey,quy2021fair}). Hence, we focus on clustering with \textit{fixed-size} clusters, where each cluster is required to have a predetermined size (see, e.g., \cite{geetha2009improved}). This reflects many real-world scenarios where lack of size constraints can result in highly skewed clusters with arbitrary sizes, where one cluster is significantly larger than the others with small sizes, thus undermining their practical usefulness. For instance, capacity constraints naturally arise from limitations such as the number of customers an individual salesperson can effectively serve.

% Capacitated clustering also sometimes refers to the presence of capacity constraints specifying an upper bound on the size of each cluster (see, e.g., \cite{Dinler2016Survey,quy2021fair})

When the number of points is restricted to be \textit{even} and each cluster must be of size \textit{exactly} two, our problem reduces to the \textit{Min-cost Perfect Matching with Delays} (MPMD) problem \cite{emek2016online,azar2017polylogarithmic}. MPMD is often studied under the classic \textit{worst-case (adversarial) model}, modeling agents' arrival order and locations as controlled by an \textit{adversary}, whose goal is degrading an algorithm's performance. Here, the performance of an online algorithm is often measured in terms of its \textit{competitive ratio} \cite{fiat1998online}, defined as the worst-case ratio between the algorithm’s total cost and that of an optimal offline solution. However, the worst-case model is overly pessimistic: even if the metric space is known upfront, no online algorithm has a constant competitive ratio \cite{ashlagi2017Min}. %, as the competitive ratio of any algorithm is at least $\Omega(\log m/\log\log m)$ 

\paragraph{\textbf{Contributions.}} To enable the design of  effective algorithms in practice, we consider the \textit{unknown i.i.d.} (UIID) model, where points' locations are drawn independently across time from an \textit{unknown} and \textit{fixed} probability distribution over the finite metric space. Under the UIID model, we measure the performance of an online algorithm in terms of its \textit{ratio-of-expectations} (RoE) \cite{mari2024online}, evaluating the expected total costs of an online algorithm with that of an offline optimal clustering, as the number of points increases. This asymptotic perspective captures how an algorithm performs in large-scale and practical scenarios, rather than small pathological worst-case cases that may not truly reflect its efficiency. An algorithm's RoE is at least $1$, with better performance corresponding to a smaller RoE. We offer hope for beyond worst-case adversaries: we devise a greedy algorithm with a \textit{\textbf{constant}} ratio-of-expectations. Finally, we show how our results easily extend to non-metric spaces and cases with lower and upper bounds on cluster sizes.

\section{Related Work}
% Note that when restricting to an \textit{even} number of points is and requiring each cluster to contain \textit{exactly} two points, our problem degenerates to the \textit{Min-cost Perfect Matching with Delays} (MPMD) problem \cite{emek2016online}.

\textit{Center-based} metric clustering and its variants with cluster size constraints have been extensively studied in both \textit{offline} settings \cite{ahmadian2019better,charikar1999constant,lattanzi2017consistent,mulvey1984solving,geetha2009improved}, and \textit{dynamic} settings \cite{christou2023efficient,guptaa2024capacitated,ailon2009streaming,chan2018fully}. Center-based metric clustering is well-suited for applications such as facility location \cite{arya2001local,fotakis2006incremental,gupta2016approximation,eisenstat2014facility}, where elements correspond to agents and facilities, while cluster centers are potential sites for public infrastructure (e.g., parks), with each agent preferring the facility closest to them. However, in this work we focus on \textit{non-centroid} metric clustering (see, e.g., \cite{Caragiannis2024Proportional}), where there are no cluster centers.

Closely related to our work is the correlation clustering problem introduced by \citet{bansal2004correlation}, where the input is a weighted graph so that a weight's sign indicates, for each pair of points, whether they should be placed in the same cluster. The goal is then to find a clustering that minimizes the number of disagreements with these pairwise recommendations. In \textit{offline} settings, the problem is NP-hard and admits several approximation algorithms (see, e.g., \cite{chawla2015near,giotis2006correlation,demaine2006correlation}). In \textit{online} settings under the classic worst-case model, it has been proven that the competitive ratio of any online algorithm is at least $\Omega(n)$ for a graph with $n$ nodes. To escape this impossibility, online correlation clustering has been explored in either \textit{stochastic} settings \cite{ailon2008aggregating,lattanzi2017consistent} or scenarios that allow re-assignments at a cost \cite{cohen2022online}, where algorithms with a constant competitive ratio have been devised. Yet, online non-centroid clustering in \textit{metric} spaces has been largely overlooked. Further, prior works on online clustering typically assume that assignments to clusters are \textit{immediate} upon arrival. To the best of our knowledge, online non-centroid clustering with \textit{delays} has not been studied before.

% Highly tied to our research are \textit{online matchings}, which can be viewed as a \textit{constrained} version of our settings where cluster sizes are at most $2$. 
As noted earlier, the \textit{Min-cost Perfect Matching with Delays} (MPMD) problem presented by \citet{emek2016online} is a special case of our setting. In the worst-case model with a finite metric space of size $m$ that is \textit{known} upfront, the best known competitiveness is $\mathcal{O}(\log m)$ \cite{azar2017polylogarithmic}, and any algorithm's competitive ratio is at least $\Omega(\log m/\log\log m)$ \cite{ashlagi2017Min}. If the metric space is \textit{unknown}, the current best competitive ratio is $\mathcal{O}(n^{\log(1.5)+\varepsilon}/\varepsilon)$ for $n$ online arrivals and $\varepsilon>0$. While MPMD has been studied under a Poisson arrival model with \textit{known} arrival rates \cite{mari2025online}, we consider the more challenging and general \textit{unknown i.i.d. arrival model}, where neither the distribution, its form nor its parameters are known. Our problem is more closely related to $k$-way MPMD, where elements must be matched into sets of size \textit{exactly} $k$. In the worst-case model, any \textit{randomized} algorithm has a generally \textit{unbounded} competitive ratio \cite{Kawase2025online,kakimura2025deterministic}. Yet, our setting is more challenging, allowing clusters of different sizes. We also study the \textit{unknown i.i.d. model}, giving hope for beyond worst-case adversaries by devising an algorithm with \textit{\textbf{constant}} competitiveness. 

Our work captures cases where agents prefer to be close to others in their cluster, exemplifying \textit{coalition formation}, where \textit{agents} collaborate in \textit{coalitions} instead of acting alone. A popular model for studying coalition formation is \textit{hedonic games} \cite{dreze1980hedonic}, where agents express preferences for coalitions they belong to by disregarding externalities. Hedonic games have been extensively studied in \textit{offline} settings (see, e.g., \cite{aziz_savani_moulin_2016,woeginger2013core,cohen2023complexity}), where settings with \textit{bounded} coalition sizes \cite{monaco2023nash,levinger2024coalition,fioravantes2025exact} and \textit{fixed-size} coalitions \cite{bilo2022hedonic,cseh2019pareto} have recently received increased attention. Yet, the above works consider hedonic games in \textit{offline} settings, while we study \textit{online} non-centroid clustering problems with clusters of a fixed and predetermined size. Recently, an \textit{online} variant of hedonic games has been explored \cite{cohen2024online,cohen2024onlinefriends,cohen2023online,cohen2025online,cohen2025decentralized,cohen2025fair,flammini2021online,bullinger2025online}. However, those studies assume that decisions are \textit{immediate}, while we allow postponing them at a cost. \textbf{In Section \ref{sec:Online Additively Separable Hedonic Games with Agent Types}, we explain how our results also extend to hedonic games.}

\section{Preliminaries}
\label{sec:Preliminaries}

We consider the \textit{Online non-centroid Clustering with Delays} (OCD) problem, where $n$ points, representing elements like agents or data, arrive one at a time over multiple rounds in a finite metric space. Formally, the input to OCD is a sequence $\sigma$ of $n$ points that arrive online over multiple rounds in a finite \textit{metric space} $\mathcal{M} = (\mathcal{X},d)$ with a set $\mathcal{X}$ of $|\mathcal{X}|<\infty$ different locations equipped with a distance function $d: \mathcal{X} \times \mathcal{X} \rightarrow \mathbb{R}^+$ that satisfies the triangle inequality. Each point $i \in \sigma$ is given by its \textit{location} $\ell_i \in \mathcal{X}$ and \textit{arrival time} $t_i \in \mathbb{N}^+$. Note that points indexed consecutively in the ordering $\sigma$ do not necessarily arrive in consecutive rounds, i.e., there may exist distinct points $i,i+1 \in \sigma$ such that $t_{i+1}-t_i \geq 2$, and all points arrive after exactly $T:=\max_{i\in\sigma} t_i$ rounds. Without loss of generality, no two points arrive at the same time. For $k \in \mathbb{N}$, we hereafter denote $[k] := \{1,\dots,k\}$. To generate the sequence $\sigma$ of $n$ points, we consider the \textit{unknown i.i.d.} (UIID) model, where arrivals are sampled independently across time from a \textit{unknown}, \textit{fixed} probability distribution $\{p_{x}\}_{x \in \mathcal{X}}$ satisfying $\sum_{x \in \mathcal{X}} p_{x} \leq 1$. Specifically, each point $i$ arrives at location $x \in \mathcal{X}$ with probability $\mathbb{P}[\ell_i=x] = p_{x}$. However, no point arrives between any two consecutive arrivals of points with probability $1-\sum_{x \in \mathcal{X}} p_{x}$, which is positive if $\sum_{x \in \mathcal{X}} p_{x} < 1$. We denote by $N^t$ the set of points that arrive until time $t \in [T]$, and we set $N := N^T = N^t$ for any time $t \geq T$. %When the distribution is \textit{stationary}, i.e., $p_{x,t} = p_x$ at each time $t \in [T]$, we obtain the \textit{unknown i.i.d.} (UIID) setting \cite{bansal2012lp}.

At each time $t \in [T]$, an online algorithm $\mathcal{A}$ shall produce a \textit{partial clustering} $\mathcal{C}^t$ of the points \textit{which arrived until time $t$} into disjoint subsets of points (i.e., \textit{clusters}), without any knowledge about future points. Once point $i$ arrives at time $t_i$, its location $\ell_i$ is revealed. Then, $\mathcal{A}$ has to \textit{irrevocably} assign point $i$ to an existing cluster in $\mathcal{C}^{t-1}$ or create a new cluster $\{t\}$, if possible. Unlike prior research on online clustering and online coalition formation \cite{cohen2024online,cohen2024onlinefriends,cohen2023online,cohen2025online,cohen2025decentralized,cohen2025fair,flammini2021online,bullinger2025online,cohen2022online}, in our work an online algorithm $\mathcal{A}$ does \textit{not} have to \textit{immediately} assign an arriving point $i$ to a cluster, but it may postpone the decision to some time $s_i \geq t_i$ with a \textit{\textbf{delay cost}}, which equals point $i$'s waiting time, denoted as $w_i := s_i-t_i$. %Namely, at time $t=s_i$, the algorithm $\mathcal{A}$ \textit{irrevocably} decides whether to insert point $i$ to an existing cluster $\mathcal{C}^{t-1}$ or create a new cluster $\{i\}$ containing, if possible. In both cases, the algorithm suffers a \textit{delay cost} of $w_i$. 
However, if no assignment is made at time $t\in [T]$, then $\mathcal{C}^{t}=\mathcal{C}^{t-1}$. Moreover, if points remain unassigned once the last point arrives (i.e., at time $t_n$), then they will keep waiting until the algorithm inserts them into clusters. We thus denote the number of clusters in $\mathcal{C}^t$ as $|\mathcal{C}^t|$ and the cluster in $\mathcal{C}^t$ containing point $i \in N^t$ as $\mathcal{C}^t(i)$.  

Once two distinct points $i \neq j$ are in the same cluster, a \textit{\textbf{connection cost}} is incurred, given by the distance between them plus their delay costs, i.e., $d(\ell_i, \ell_j) + w_i + w_j$. Therefore, we evaluate the quality of the clustering $\mathcal{C}^t$ at time $t\in [T]$ by its \textit{\textbf{total cost}}, defined as the sum of the connection costs in each cluster, i.e., $\TC(\mathcal{C}^t, \mathbf{w}) = \sum_{C \in \mathcal{C}^t} \sum_{i,j \in C: i \neq j} [d(\ell_i, \ell_j) + w_i + w_j]$, where $\mathbf{w} = (w_i)_{i \in N}$ is the \textit{delay profile}. For any cluster $C \in \mathcal{C}^t$ and point $i\in C$, point $i$'s delay cost affects her connection cost to any other agent $j \in C$. This captures scenarios where a point’s delay imposes a cost on \textit{each} cluster member, not just the delayed point. For example, in online gaming, all players must be present before the game starts. Hence, the delay of one player slows down \textit{all} players’ start time, so each cluster member suffers from the waiting time of late arrivals. It is thus reasonable to view that waiting times impose cost proportional to the number of participants affected. 

We also remark that summing connection (distance) and delay (time) costs is a standard objective in the literature on online problems with delays (see, e.g., \cite{emek2016online} and subsequent works). Here, the interpretation is that both terms represent the same unit of cost. If one requires a different conversion via a user-supplied conversion function $f:\mathbb{N}\rightarrow\mathbb{R}^+$, a point $i$ with waiting time $w_i$ suffers a delay cost of $f(w_i)$. Our analysis extends to some choices of $f$, but analyzing the general case is an interesting direction for future work. For example, $f(t)=\lambda t$ with $\lambda>0$ scales time units; here, our analysis carries through with constants depending on $\lambda$ (proofs are almost identical up to this scaling). \textbf{The main goal of $\mathcal{A}$ is thus finding a \textit{cost-optimal} clustering $\mathcal{C}^\star$ that \textit{minimizes} the total cost among all possible clusterings of the $n$ points.}

% \begin{remark}[Commensurability of Distance and Time]
    
% \end{remark}

When clusters can be of any size, the clustering of singletons $(\{i\})_{i\in\sigma}$ is trivially cost-optimal, as it incurs a total cost of zero. We therefore focus on realistic cases with capacitated clusters. Namely, given $k$ positive integers $\{n_m\}_{m\in [k]}$ with $n_1\geq n_2 \geq \dots\geq n_k$ and $\sum_{m\in [k]} n_m = n$, we seek a final clustering $\mathcal{C}=(C_m)_{m\in[k]}$ where each point is assigned to one of $k$ clusters such that the size of the $m$-th cluster is exactly $|C_m|=n_m$. We require that $2 \leq k < n$ and $2 \leq n_m < n$ for any $m\in[k]$, as cases where either $k=1$, $n_m=1$, $k=n$ or $n_m=n$ are trivial. Moreover, we treat the cluster sizes $\{n_m\}_{m\in [k]}$ as \textit{constants}, while allowing the number of clusters $k$ to grow with the number of points $n$, i.e., $k$ and $n$ are not constants. This assumption is natural in several applications where cluster sizes are intrinsic to the system. In online gaming, tournament systems, and team-based e-sports platforms, team sizes (e.g., 5-player squads) are fixed and cannot be chosen by the algorithm. Similarly, in ride-sharing and package-delivery services, vehicles have predetermined capacities. Further, in balanced document clustering, practical constraints, such as workloads of downstream human reviewers, also require explicit cardinality bounds. As noted in Section \ref{sec:Lower and Upper Bounds on Cluster Sizes}, our analysis also extends to requiring lower and upper bounds on coalition sizes, further generalizing our model beyond the fixed-size assumption.

We measure the performance of an online algorithm $\mathcal{A}$ in terms of its \textit{ratio-of-expectations} \cite{garg2008stochastic}, comparing the expected cost of $\mathcal{A}$ with the expected cost of the offline cost-optimal clustering. Formally, we denote by $\mathcal{I}$ the distribution over sequences of $n$ points generated by the UIID model. For each sampled sequence $\sigma \sim \mathcal{I}$, let $\OPT(\sigma)$ and $\mathcal{A}(\sigma)$ be the total cost of the \textit{offline} cost-optimal clustering and the clustering generated by running the algorithm $\mathcal{A}$ on $\sigma$, respectively. We also denote $\opt(\mathcal{I}) := \mathbb{E}[\OPT(\sigma)]$ and $\mathcal{A}(\mathcal{I}) := \mathbb{E}[\mathcal{A}(\sigma)]$, where $\mathbb{E}[\cdot]$ is the expectation over $\sigma \sim \mathcal{I}$ and the possible randomness of the algorithm. We thus say that $\mathcal{A}$ is $c$-competitive for $c \geq 1$ under the \textit{ratio-of-expectations} (RoE) if:
\begin{equation}
    \label{eq:roe}
    \roe(\mathcal{A}):=\overline{\lim}_{n \rightarrow \infty} \frac{\mathcal{A}(\mathcal{I})}{\opt(\mathcal{I})} \leq c
\end{equation}
where greater performance is indicated by a smaller RoE. The RoE measures $\mathcal{A}$'s asymptotic behavior and robustness as the number of points grows, thus evaluated at the limit $n \rightarrow \infty$. This highlights $\mathcal{A}$'s efficiency in practical, large-scale scenarios instead of small pathological cases that may not truly reflect $\mathcal{A}$'s effectiveness. %When no ambiguity arises, we omit $\mathcal{I}$.

\section{A \textit{Constant} Competitive Algorithm}
\label{sec:A Constant Competitive Greedy Algorithm}
We now present the deterministic and greedy Algorithm \ref{alg:delayed greedy}, termed as \textit{Delayed Greedy} (\textsc{DGreedy}), which has a \textbf{\textit{constant}} ratio-of-expectations (RoE) and runs as follows. At each time $t$, if no point arrives at time $t$, we proceed to treating pending points (line \ref{state:if next round}-\ref{state:next round}). Specifically, for any pending point $i$, we identify two sets of candidate clusters to which point $i$ can be assigned. First, we construct $\mathcal{S}_{i}$, composed of any \textit{existing} cluster $C_m \in \mathcal{C}^t$ for some $m\in [k]$ such that $|C_m |\leq n_m$ and point $i$'s distance from any point $j \in C_m$ is at most point $i$'s waiting time plus point $j$'s delay cost (lines \ref{state:existing2}). Afterwards, if there are at least two pending points, then we build $\mathcal{D}_{i}$, consisting of any pair $(j,m)\in (N^t\setminus\{i\}) \times [k]$ such that $C_m=\emptyset$ and $j$ is a pending point for which the total waiting time of points $i,j$ exceeds the distance between them (lines \ref{state:if pending}-\ref{state:pending last}). 

\begin{algorithm}[tb]
    \caption{\textsc{DGreedy}}
    \label{alg:delayed greedy}
    \textbf{Input}: A sequence $\sigma$ of $n$ points; Cluster sizes $\{n_m\}_{m\in [k]}$.
    \begin{algorithmic}[1] %[1] enables line numbers
        \State{Initialize an empty clustering $\mathcal{C} \leftarrow (C_m=\emptyset)_{m\in[k]}$.}
        \For{$t = 1, 2, \dots, T$}
            \If{no point arrives at time $t$\label{state:if next round}} 
                \State{Proceed to treating pending points. \label{state:next round}}
            \EndIf
            \For{any pending point $i$} 
                \State{Set $\mathcal{S}_{i} \leftarrow \{m\in[k]:|C_m|<n_m \text{ and } \forall j \in C_m \text{ s.t. } d(\ell_{i}, \ell_{j}) \leq (t-t_i) + w_{j}\}$.\label{state:existing2}}  
                \If{there are at least two pending points\label{state:if pending}}
                    \State{Set $\mathcal{D}_{i} \leftarrow \{ (j,m): i\neq j\in N^t \text{ is pending, }  d(\ell_{i}, \ell_{j}) \leq (t-t_i)+(t-t_j) \text{ and }C_m = \emptyset\}$.\label{state:pending last}}
                \EndIf 
                \If{$|\mathcal{S}_{i}| > 0$} \Comment{\textbf{Best existing cluster}\label{state:existing non empty}}
                    \State{Breaking ties arbitrarily, pick $m_1 \in \argmin_{m \in \mathcal{S}_{i}} \TC(\mathcal{C}_{+\{i\}\rightarrow m})$, where $\mathcal{C}_{+C'\rightarrow m}$ is the clustering obtained from $\mathcal{C}$ by adding points in some $C' \subseteq N$ to  $C_{m}$.\label{state:existing best}}
                \Else { $m_1=0$ \label{state:existing empty}}
                \EndIf
                \If{$|\mathcal{D}_{i}| > 0$} \Comment{\textbf{Best pending pair}\label{state:pending non empty}}
                    \State{$\mathcal{D}_{i}^\star = \argmin_{(j,m) \in \mathcal{D}_{i}} \TC(\mathcal{C}_{+\{i,j\}\rightarrow m})$.}
                    \State{\textbf{Pick a pair with maximum total waiting time:}\label{state:best pending max}}
                    \State{Select $(j,m_2) \in \argmax_{(j,m) \in \mathcal{D}_{i}^\star} (t-t_j)$ (ties broken arbitrarily).\label{state:best pending last}} 
                \Else { $m_2=0$ \label{state:pending empty}}
                \EndIf 
                
                \If{$m_2 \neq 0$ \textbf{and} ($m_1 = 0$ \textbf{or} ($m_1 \neq 0$ \textbf{and}  $\TC(\mathcal{C}_{+\{i,j\}\rightarrow m_2}) \leq \TC(\mathcal{C}_{+\{i\}\rightarrow m_1})$))\label{state:assign to pending}} 
                \State{Add $i,j$ to the currently empty cluster $C_{m_2}$.\label{state:assign to pending last}}
                \ElsIf{$m_1 \neq 0$ \textbf{and} $m_2 = 0$ \label{state:assign to existing}}
                    \State{Add point $i$ to the existing cluster $C_{m_1}$.\label{state:assign to existing last}}
                \EndIf

            \EndFor 
            
        \EndFor
    \end{algorithmic}
\end{algorithm}

If $|\mathcal{S}_{i}| >0$, then we pick $m_1\in\mathcal{S}_{i}$ such that inserting point $i$ to the cluster $C_{m_1}$ \textit{minimizes} the increase in the total cost of the current clustering (lines \ref{state:existing non empty}-\ref{state:existing best}); otherwise, we set $m_1=0$ (line \ref{state:existing empty}). Similarly, if $|\mathcal{D}_i| > 0$, then we select a pair $(j,m_2) \in \mathcal{D}_i$ such that forming a new cluster $\{i,j\}$ by putting points $i,j$ in the currently empty $m_2$-th cluster $C_{m_2}$ \textit{minimizes} the increase in the total cost of the current clustering (line \ref{state:pending non empty}). Among all such options, we pick the one that \textit{maximizes} the total waiting time of its points (lines \ref{state:best pending max}-\ref{state:best pending last}); otherwise, we set $m_2=0$ (line \ref{state:pending empty}). Intuitively, by breaking ties based on the total waiting time, we balance immediate connection costs with long-term delay costs, ensuring that highly delayed points are treated sooner, preventing their costs from growing even further. 

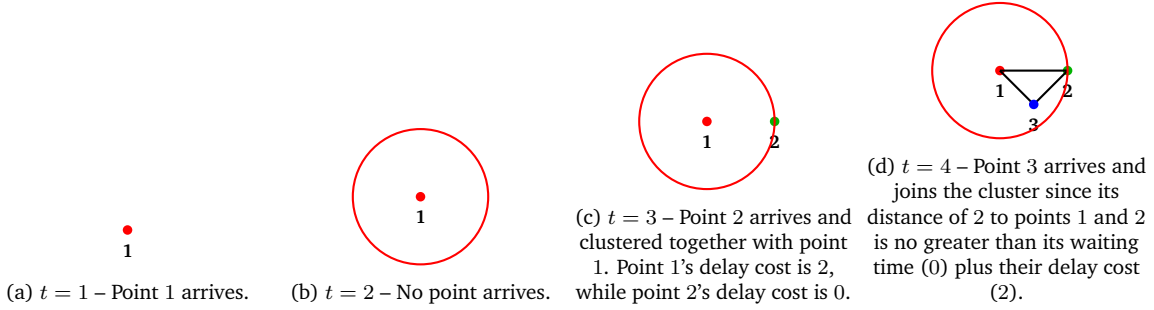
\begin{figure*}[t!]
\centering

\newcommand{\point}[4]{
  \fill[#4] (#1, #2) circle (2pt);
  \node at (#1, #2 - 0.3) {\scriptsize #3};
}

\newcommand{\expandingball}[4]{
  \draw[#4, thick] (#1, #2) circle (#3);
}

\begin{subfigure}[b]{0.23\textwidth}
\centering
\begin{tikzpicture}[scale=0.9]
  % Time 0: Only point 1
  \point{0}{1}{\textbf{1}}{red}
  % \expandingball{0}{2}{0.3}{red}
\end{tikzpicture}
\caption{\centering $t=1$ -- Point $1$ arrives.}
\end{subfigure}
\begin{subfigure}[b]{0.23\textwidth}
\centering
\begin{tikzpicture}[scale=0.9]
  % Time 0: Only point 1
  \point{0}{1}{\textbf{1}}{red}
  \expandingball{0}{1}{1}{red}
\end{tikzpicture}
\caption{\centering $t=2$ -- No point arrives.}
\end{subfigure}
\begin{subfigure}[b]{0.23\textwidth}
\centering
\begin{tikzpicture}[scale=0.9]
  % Time 3: point 2 arrives
  \point{0}{1}{\textbf{1}}{red}
  \point{1}{1}{\textbf{2}}{green!70!black}
  \expandingball{0}{1}{1}{red}
  % \expandingball{1.5}{2}{0.3}{green!70!black}
\end{tikzpicture}
\caption{\centering $t=3$ -- Point $2$ arrives and clustered together with point $1$. Point $1$'s delay cost is $2$, while point $2$'s delay cost is $0$.}
\end{subfigure}
\begin{subfigure}[b]{0.23\textwidth}
\centering
\begin{tikzpicture}[scale=0.9]
  % Time 10: point 3 arrives
  \point{0}{1}{\textbf{1}}{red}
  \point{1}{1}{\textbf{2}}{green!70!black}
  \draw[thick] (0,1) -- (1,1);
  \draw[thick] (0.5,0.5) -- (1,1);
  \draw[thick] (0,1) -- (0.5,0.5);
  \point{0.5}{0.5}{\textbf{3}}{blue}
  \expandingball{0}{1}{1}{red}
  % \expandingball{0}{0}{0.3}{blue}
\end{tikzpicture}
\caption{\centering $t=4$ -- Point $3$ arrives and joins the cluster since its distance of $2$ to points $1$ and $2$ is no greater than its waiting time ($0$) plus their delay cost ($2$).}
\end{subfigure}

\caption{An example of how \textsc{DGreedy} works for clusterings with a single cluster of size $3$ on a sequence of three points $1, 2, 3$ arriving at times $1, 3, 4$ in a finite metric space consisting of $3$ locations, where the distance between each pair of locations is $2$.}
\label{fig:greedy}
\end{figure*}

If we indeed picked a pair $(j,m_2) \in \mathcal{D}_i$ as above (i.e., $m_2 \neq 0$), then we insert points $i,j$ into the currently empty cluster $C_{m_2}$ if either no existing cluster with minimum increase in total cost was also selected (i.e., $m_1=0$), or an existing cluster $C_{m_1}$ was found and the increase in total cost incurred by forming the new cluster $\{i,j\}$ is at most that of adding point $i$ to $C_{m_1}$ (lines \ref{state:assign to pending}-\ref{state:assign to pending last}). Otherwise, if $m_1\neq 0$ and $m_2=0$, we add point $i$ to the existing cluster $C_{m_1}$ (lines \ref{state:assign to existing}-\ref{state:assign to existing last}). In any other case, point $i$ will continue to wait and we proceed to the next pending point.

Note that the algorithm is well-defined. Throughout its execution, \textsc{DGreedy} enforces that the $m$-th cluster $C_m$ has size at most $n_m$ for each $m\in [k]$, either by adding points to $C_m$ only when it contains fewer than $n_m$ members, or by forming new clusters only of size $2$ by assigning two pending points to an empty cluster. As $\sum_{m\in [k]} n_m = n$, every point is eventually assigned to a cluster since the metric space $\mathcal{M}$ contains a \textit{finite} number of points, yielding that the \textit{last} point's waiting time is bounded by the diameter of $\mathcal{M}$. We formulate this in the following observation:

\begin{observation}
    \label{obs:size k}
    For any sequence of points $\sigma$, the final clusters generated by \textsc{DGreedy} contain exactly $k$ clusters, where the $m$-th cluster is of size exactly $n_m$ for any $m \in [k]$.
\end{observation}

\paragraph{\textbf{Geometric Interpretation}.} To better grasp our algorithm, we consider its geometric interpretation (See Figure \ref{fig:greedy} for an illustration). When point $i$ arrives, a ball centered at its location $\ell_i$ expands uniformly over time and stops growing when point $i$ is assigned to a cluster. The ball's radius represents the delay cost for leaving point $i$ unassigned. Thus, once the ball of a pending point $i$ intersects that of either some point within an existing cluster or another pending point, point $i$ is assigned to such an intersecting cluster.

Next, we first provide in an upper bound on the (expected) total cost of the clustering generated by \textsc{DGreedy} (Section \ref{sec:upper bound delayed greedy}). After devising a lower bound on the (expected) total cost of the optimal (offline) clustering (Section \ref{sec:Lower Bounding the Optimal total cost}), we infer \textsc{DGreedy}'s \textit{\textbf{constant}} ratio-of-expectations (Section \ref{sec:Ratio-of-Expectations greedy}).

\subsection{Upper Bounding \textsc{DGreedy}'s total cost}
\label{sec:upper bound delayed greedy}
In this section, we obtain an upper bound on the expected total cost of the clustering produced by \textsc{DGreedy}, and begin with a general bound on the incurred total cost:
\begin{lemma}
    \label{lemma:at most twice the delay cost}
    Given a {\normalfont finite} metric space $\mathcal{M}=(\mathcal{X},d)$, for any sequence of $n$ points $\sigma$, the total cost of the final clustering $\mathcal{C}$ generated by \textsc{DGreedy} satisfies:
    $$\TC(\mathcal{C}, \mathbf{w}) \leq 2 (n_1-1) \sum_{i \in \sigma} w_i$$
\end{lemma}
\begin{proof}
    For any pair of points $i$ and $j \in \mathcal{C}(i)$ with $i\neq j$, as their connection cost is at most the sum of their delay costs, i.e.,  $d(\ell_{i}, \ell_{j}) \leq w_i + w_{j}$, we obtain: 
    \begin{equation*}
        \begin{aligned}
            \TC(\mathcal{C}, \mathbf{w}) &= \sum_{C \in \mathcal{C}} \sum_{i, j \in C: i \neq j} [d(\ell_{i}, \ell_{j}) + w_i + w_j] \\
            &\leq \sum_{C \in \mathcal{C}} \sum_{i, j \in C: i \neq j} 2[w_i + w_{j}] \\
            &= 2 \sum_{C \in \mathcal{C}} (|C|-1) \sum_{i \in C} w_i \leq 2 (n_1-1) \sum_{i \in \sigma} w_i
        \end{aligned}
    \end{equation*}
    where the last inequality is by $|C|\leq n_1$ for any cluster $C \in \mathcal{C}$ (Observation \ref{obs:size k}) and the $m$-th cluster's size satisfies $n_m \leq n_1$. 
\end{proof}

Hence, by finding an upper bound on the expected total delay cost, we can estimate the overall expected total cost of the clustering produced by \textsc{DGreedy}. For this sake, we next bound the delay cost of each point. For any subset of locations $\mathcal{Y} \subseteq \mathcal{X}$, if a point arrived at some location in $\mathcal{Y}$ during time $t$, then we next make the following observation regarding the \textit{\textbf{waiting time}} $\tau_t^\mathcal{Y}$ for another point to arrive at some location in $\mathcal{Y}$:
\begin{observation}    
    \label{obs:expected waiting time between arrivals}
    In the UIID model with a {\normalfont finite} metric space $\mathcal{M}=(\mathcal{X},d)$, the waiting time $\tau_t^\mathcal{Y}$ is geometrically distributed with a success probability of $q_{\mathcal{Y}} := \sum_{y \in \mathcal{Y}} p_{y}$, and thus its expectation in the UIID model is $\mathbb{E}[\tau_t^\mathcal{Y}]= \frac{1}{q_{\mathcal{Y}}}$.
\end{observation}

To bound  each point's delay cost, we distinguish between two types of points that may arrive after it. First, for any possible location $x \in \mathcal{X}$, let $\bar{B}(x,r)$ (resp. $B(x,r)$) be the \textit{closed} (resp. \textit{open}) ball centered at $x$ with radius $r > 0$, i.e., the set of locations $y$ with $d(x,y) \leq r$ (resp. $d(x,y) < r$). A point arriving at location $x$ is said to be assigned to a cluster together with a \textit{close point} if the distance between them is at most location $x$'s radius $r_x$, which is defined as:
\begin{equation}
    \label{eq:radius at locations x}
    \begin{aligned}
        r_x := \min \left\{r \geq 0: \frac{1}{\sum_{y \in \bar{B}(x,r)} p_{y}} \leq r\right\}
    \end{aligned}
\end{equation}
Intuitively, the choice of $r_x$ balances the expected waiting time $\mathbb{E}[\tau_t^{\bar{B}(x,r)}]$ between two arrivals within the ball $\bar{B}(x,r)$ (by Observation \ref{obs:expected waiting time between arrivals}) and this ball's diameter. It is also well-defined since the function $r \mapsto \mathbb{E}[\tau_t^{\bar{B}(x,r)}]$ is non-increasing, and thus $r_x \in (0, 1/p_x]$. 

For any point $i$, we say that it is an \textit{\textbf{early point}} if $i$ is not pending at time $t_i+r_{\ell_i}$ and there is another point $j$ with $t_j-t_i > r_{\ell_i}$ such that $d(\ell_i, \ell_{j}) \leq r_{\ell_i}$; otherwise, we say that point $i$ is a \textit{\textbf{late point}}. For each \textit{early} point $i$, we set:
\begin{equation*}
         \alpha_i^{\text{early}} := \begin{cases}
        0, \text{ if $i$ is assigned to a cluster at time $t_i + r_{\ell_i}$}\\
        \min_{j \in N} \{t_j-t_i - r_{\ell_i}: t_j-t_i > r_{\ell_i} \text{ and } d(\ell_i, \ell_{j}) \leq r_{\ell_i}\},\text{ otherwise}
    \end{cases}
\end{equation*}
Now, we upper bound the waiting time of the \textit{\textbf{early}} point $i$:
\begin{lemma}
    \label{lemma:waiting time of early points}
    In the UIID model with a {\normalfont finite} metric space $\mathcal{M}=(\mathcal{X},d)$, for any sequence of $n$ points $\sigma$ and each {\normalfont early} point $i\in\sigma$, it holds that $w_i \leq r_{\ell_i} + \alpha_i^{\text{early}}$.
\end{lemma}
\begin{proof}
    Let $i$ be an early point. Let $j$ be the \textit{first} point with $t_j-t_i > r_{\ell_i}$ such that $d(\ell_i, \ell_{j}) \leq r_{\ell_i}$, i.e., $t_j-t_i = r_{\ell_i}+\alpha_i^{\text{early}}$. If point $i$ has already been assigned to a cluster at time $t_j$, then $s_i \leq t_j$ and $w_i = s_i - t_i \leq t_j - t_i = r_{\ell_i}+\alpha_i^{\text{early}}$. Otherwise, we show that a cluster containing points $i, j$ is created at time $t_j$. Note that $(t_j-t_j)+(t_j-t_i) = r_{\ell_i}+\alpha_i^{\text{early}} > r_{\ell_i} \geq d(\ell_i, \ell_{j})$. As points $i$ and $j$ are pending points at time $t_j$ and the greedy criteria is met, then clusters of pending points that contain $i, j$ are considered eligible by \textsc{DGreedy}. Assume, towards contradiction, that there is another pending point $j'$ such that a cluster containing points $j$ and $j'$ has already been formed at time $t_j$, i.e., $d(\ell_{j},\ell_{j'}) \leq t_j-t_{j'}$. By the triangle inequality, we have: $(t_j-t_i)+(t_j-t_{j'}) > d(\ell_i, \ell_{j}) + d(\ell_{j},\ell_{j'}) \geq d(\ell_{i},\ell_{j'})$. That is, a cluster containing the points $i, j'$ should have been created before the arrival of point $j$, which is a contradiction. We infer that $w_i = s_i - t_i = r_{\ell_i}+\alpha_i^{\text{early}}$.
\end{proof}

In the following, we bound the total delay cost of \textit{\textbf{late}} points. Sadly, the waiting time of a late point may be the highest value possible, i.e., the diameter $d_{\max} := \max_{x,y \in \mathcal{X}} d(x,y)$ of the metric space $\mathcal{M}$. However, we prove that only a few such points exist. Recalling that $t_n$ is the arrival time of the last point, then we define: 
\begin{equation*}
        \alpha_i^{\text{late}} := \begin{cases}
        0 , \text{ if $t_n \leq t_i + d_{\max}$}\\
        \min_{j \in N} \{t_j - (t_i+d_{\max}) : t_j > t_i + d_{\max}\} , \text{ otherwise}
    \end{cases}
\end{equation*}
Therefore, we obtain that:
\begin{lemma}
    \label{lemma:waiting time of late points}
    In the UIID model with a {\normalfont finite} metric space $\mathcal{M}=(\mathcal{X},d)$, for any sequence $\sigma$ and each location $x \in \mathcal{X}$:
    \begin{enumerate}
        \item There is at most one {\normalfont late} point at $x$, i.e., there are at most $|\mathcal{X}|<\infty$ late points. 
        \item For each {\normalfont late} point $i$, it holds that $w_i \leq d_{\max} + \alpha_i^{\text{late}}$.
    \end{enumerate}
\end{lemma}
\begin{proof}
    We prove each part separately:
    \begin{enumerate}
        \item Consider two points $i,j$ with $\ell_i = \ell_j$ and $t_i < t_j$. Assume, towards contradiction, that points $i$ and $j$ are late. By the definition of a late point, we have that $t_j-t_i\leq r_{\ell_i}$ and point $i$ remains unassigned to a cluster at time $t_j$. As point $j$ is also late, it is not assigned to a cluster upon its arrival and \textsc{DGreedy} does not form the cluster $\{i,j\}$. However, $0 = d(\ell_i, \ell_j) \leq r_{\ell_i} + r_{\ell_j}$, and thus \textsc{DGreedy} should create the cluster $\{i, j\}$ at time $t_j$, which is a contradiction. As such, point $i$ cannot remain a pending point so long as at least one point arrives at its location, yielding that it is not a late point. That is, at most one late point is located at each location $x \in \mathcal{X}$, i.e., there are at most $|\mathcal{X}|<\infty$ late points (as the metric space is \textit{finite}). 

        \item Given a late point $i$, we distinguish between the following:
        \begin{enumerate}
            \item If $t_n \leq t_i + d_{\max}$, then all the points have already arrived by time $t=t_i+d_{\max}$. Thus, either point $i$ has already been assigned to a cluster or not. In the latter case, the following cases are possible:
                 \begin{enumerate}
                    \item If all clusters are non-empty when \textsc{DGreedy} processes point $i$, then there is at least one cluster of index $m \in [k]$ such that $|C_{m'}|< n_{m'}$, as the $m$-th cluster should be eventually of size $n_m$ and $\sum_{m\in [k]} n_m = n$. Since $t-t_i = d_{\max} \geq d(\ell_i, \ell_j)$ for any point $j$ in the $m$-th cluster $C_m$ for every $m\in [k]$, our algorithm assigns point $i$ to such a cluster at time $t$, yielding $w_i = d_{\max}$.
                            
                    \item Otherwise, when our algorithm processes point $i$, there exists at least one empty cluster of index $m \in [k]$, i.e., $C_m = \emptyset$. Since the $m$-th cluster should eventually contain exactly $n_m \geq 2$ points, there is at least one other pending point $j\neq i$ at time $t$. Because $t-t_i = d_{\max} \geq d(\ell_i, \ell_j)$ for any other pending point $j\neq i$, a cluster of pending points containing $i, j$ is considered eligible by \textsc{DGreedy} at time $t$ when point $i$ is being processed by \textsc{DGreedy}. At this time, if all non-empty clusters are full, then our algorithm assigns point $i$ together with other pending points to a currently empty cluster, meaning that $w_i = d_{\max}$. Otherwise, when our algorithm processes point $i$, there exists at least one partially filled cluster of index $m' \in [k]$ with $1 \leq |C_{m'}| < n_{m'}$. For any such $m' \in [k]$, since $t-t_i = d_{\max} \geq d(\ell_i, \ell_j)$ for any point $i\neq j \in C_{m'}$, then the algorithm identifies a cluster $C_{m_1}$ with $1 \leq |C_{m_1}| < n_{m_1}$ such that inserting point $i$ to the cluster $C_{m_1}$ minimizes the increase in the total cost of the current clustering. Thus, at time $t$, \textsc{DGreedy} will assign point $i$ to either the $m_1$-th cluster or a currently empty cluster, depending on which yields a lower total cost. In both cases, $w_i = d_{\max}$.
                \end{enumerate}
        
            \item Otherwise, if $t_n > t_i + d_{\max}$, then the point $j$ that arrives right after point $i$ appears at time $t = t_i + d_{\max} + \alpha_i^{\text{late}}$ by the definition of $\alpha_i^{\text{late}}$. If point $i$ has already been assigned to a cluster at time $t$, then its waiting time is at most $w_i \leq d_{\max} + \alpha_i^{\text{late}}$. Otherwise, note that $t-t_i = d_{\max} + \alpha_i^{\text{late}} > d_{\max} \geq d(\ell_i, \ell_j)$. As points $i$ and $j$ are pending points at time $t$ and the greedy criteria is met, then a cluster of pending points containing $i, j$ is considered eligible by \textsc{DGreedy}. By arguments similar to the proof of Lemma \ref{lemma:waiting time of early points}, such a cluster is created. In any case, we conclude that $w_i \leq d_{\max} + \alpha_i^{\text{late}}$, as desired.
        \end{enumerate}
    \end{enumerate}
\end{proof}

We are now ready to derive an upper bound on the expected total cost of the clustering generated by \textsc{DGreedy}, which will later enable us to establish \textsc{DGreedy}'s \textit{\textbf{constant}} ratio-of-expectations in Section \ref{sec:Ratio-of-Expectations greedy}.
\begin{theorem}
    \label{thm:upper bound on the greedy expected total cost}
    In the UIID model with a {\normalfont finite} metric space $\mathcal{M}=(\mathcal{X},d)$, \textsc{DGreedy}'s expected total cost over all sequences of $n$ points sampled from a distribution $\mathcal{I}$ satisfies: 
    \begin{equation*}
        % \begin{array}{c}
            \text{\textsc{DGreedy}}(\mathcal{I}) \leq 
            2(n_1-1) \left[n \sum_{x \in \mathcal{X}} p_x r_{x}  +|\mathcal{X}|\cdot d_{\max}\right] + \frac{2(n_1-1)|\mathcal{X}|}{\sum_{x \in \mathcal{X}} p_{x}}
        % \end{array}
    \end{equation*}
    % \begin{equation}
    %     \label{eq:upper bound on the greedy expected total cost}
    %     \text{\textsc{DGreedy}}(\mathcal{I}) \leq 2(n^3-n^2+n) \sum_{x \in \mathcal{X}} p_x r_{x} - \frac{|\mathcal{X}|(n^2-n+1)}{\sum_{x \in \mathcal{X}} p_{x}}
    % \end{equation}
\end{theorem}
\begin{proof}
    First, we bound the expected delay cost of each \textit{late} point $i$. If $t_n > t_i + d_{\max}$, then the point $j$ that arrives right after point $i$ appears at time $t = t_i + d_{\max} + \alpha_i^{\text{late}}$ by the definition of $\alpha_i^{\text{late}}$. Namely, the waiting time between the arrivals of points $i$ and $j$ is $\tau_{t_i}^\mathcal{X} = t-t_i = d_{\max} + \alpha_i^{\text{late}}$. %If $\sum_{x \in \mathcal{X}} p_{x} =1$, then the waiting time is $0$, i.e., $\tau_{t_i}^\mathcal{X}=0$, yielding that $\mathbb{E}[\alpha_i^{\text{late}} | i \text{ is a late point}]=- d_{\max}$. 
    Since $\sum_{x \in \mathcal{X}} p_{x} < 1$, then by Observation \ref{obs:expected waiting time between arrivals} and since $\alpha_i^{\text{late}} = 0$ if $t_n \leq t_i + d_{\max}$, we infer that:
    \begin{equation}
        \label{eq:expected alpha late}
        % \begin{aligned}
            \mathbb{E}[\alpha_i^{\text{late}} | i \text{ is a late point}] \leq\mathbb{E}[\alpha_i^{\text{late}} | i \text{ is a late point and } t_n > t_i + d_{\max}]=-d_{\max} + \frac{1}{\sum_{x \in \mathcal{X}} p_{x}}
        % \end{aligned}
    \end{equation}    
    Since the waiting time of a late point is at most $d_{\max}$: 
    \begin{equation}
        \label{eq:waiting time upper bound}
        \begin{aligned}
            \mathbb{E}[w_i | i \text{ is a late point}] &\leq \mathbb{E}[w_i | i \text{ is a late point and } t_n > t_i + d_{\max}]+ \mathbb{E}[w_i | i \text{ is a late point and } t_n \leq t_i + d_{\max}] \\
            &\leq \mathbb{E}[w_i | i \text{ is a late point and } t_n > t_i + d_{\max}] + d_{\max}
        \end{aligned}
    \end{equation}
    Combining \eqref{eq:expected alpha late} and \eqref{eq:waiting time upper bound} with Lemma \ref{lemma:waiting time of late points}, $\mathbb{E}[w_i | i \text{ is a late point}] \leq d_{\max}+ 1/\sum_{x \in \mathcal{X}} p_{x}$ due to $\sum_{x \in \mathcal{X}} p_{x}<1$. %; otherwise, if $\sum_{x \in \mathcal{X}} p_{x}=1$, then $\mathbb{E}[w_i | i \text{ is a late point}] \leq d_{\max}$. 
    Since we consider finite metric spaces and there are at most $|\mathcal{X}|<\infty$ late points by Lemma \ref{lemma:waiting time of late points}, the expected total delay cost of the late points satisfies: 
    \begin{equation}
        \label{eq:total delay cost late points}
        % \begin{array}{c}
             \mathbb{E}\left[\sum_{i \in N : i \text{ is late}} w_i\right] \leq |\mathcal{X}|\left[d_{\max}+ \frac{1}{\sum_{x \in \mathcal{X}} p_{x}}\right]
        % \end{array}
    \end{equation}

    Next, we analyze the expected delay cost of an \textit{early} point $i$. Let $j$ be the \textit{first} point with $t_j-t_i > r_{\ell_i}$ such that $d(\ell_i, \ell_{j}) \leq r_{\ell_i}$. Namely, the waiting time between the arrivals of points $i$ and $j$ within $\bar{B}(\ell_i, r_{\ell_i})$ is $\tau_{t_i}^{\bar{B}(\ell_i, r_{\ell_i})} = t_j-t_i = r_{\ell_i}+\alpha_i^{\text{early}}$. By Observation \ref{obs:expected waiting time between arrivals} and \eqref{eq:radius at locations x}, since $\sum_{x \in \bar{B}(\ell_i, r_{\ell_i})} p_{x}<\sum_{x \in \mathcal{X}} p_{x} < 1$, then we have: 
    \begin{equation*}
        % % \begin{array}{c}
            \mathbb{E}\left[\alpha_i^{\text{early}} \middle| i \text{ is an early point}\right] = - r_{\ell_i} + \frac{1}{\sum_{x \in \bar{B}(\ell_i, r_{\ell_i})} p_{x}} \leq - r_{\ell_i} + r_{\ell_i} = 0     
        % \end{array}    
    \end{equation*} 
    Therefore, we obtain that $\mathbb{E}[w_i | i \text{ is an early point}] \leq r_{\ell_i}$ by Lemma \ref{lemma:waiting time of early points}. Thus, the expected total delay cost of early points satisfies:  
    \begin{equation}
        \label{eq:total delay cost early points}
        \begin{aligned}
            \mathbb{E}\left[\sum_{i \in N : i \text{ is early}} w_i\right] &\leq \sum_{i \in N} \sum_{x \in \mathcal{X}} \mathbb{P}[i \text{ is early and } \ell_i=x] \cdot \mathbb{E}[w_i | i \text{ is early and } \ell_i=x] \\
            &\leq \sum_{i \in N} \sum_{x \in \mathcal{X}} p_x r_{x} = n \sum_{x \in \mathcal{X}} p_x r_{x} 
        \end{aligned}
    \end{equation}    
    where we used $r_x \leq 1/p_x$ for $p_x >0$. By Lemma \ref{lemma:at most twice the delay cost}:
    \begin{equation}
        \label{eq:degreedy UB}
        \text{\textsc{DGreedy}}(\mathcal{I}) \leq \mathbb{E}\left[2 (k-1) \sum_{i \in N} w_i\right]= 2(n_1-1)\left[\mathbb{E}\left[\sum_{i \in N : i \text{ is early}} w_{i}\right] + \mathbb{E}\left[\sum_{i \in N : i \text{ is late}} w_i\right] \right]
    \end{equation}
    Combining \eqref{eq:total delay cost early points} and \eqref{eq:degreedy UB}, we conclude the desired.
\end{proof}

\subsection{Lower Bounding the Optimal total cost}
\label{sec:Lower Bounding the Optimal total cost}
In this section, we obtain a lower bound on the total cost of the (offline) cost-optimal clustering. We begin with a general lower bound on the minimal total cost. For any sequence of points $\sigma$, we denote the \textit{minimal delay cost} of each point $i \in \sigma$ by $w_i^\sigma = s_i^{\sigma}-t_i$ and its \textit{minimum cost} as:
\begin{equation*}
    % \begin{array}{c}
        c_i(\sigma) = \min_{i \neq j \in \sigma} \{d(\ell_i, \ell_j) + w_i^\sigma + w_j^\sigma\} 
    % \end{array}
\end{equation*}
Letting $\mathbf{w}^\sigma = (w_i^\sigma)_{i \in \sigma}$ be the \textit{optimal delay profile}, we obtain the following lower bound on the minimal total cost:

\begin{lemma}
    \label{lemma:the optimal cost is at least half the overall minimum cost}
    In the UIID model with a {\normalfont finite} metric space $\mathcal{M}=(\mathcal{X},d)$, for any sequence of $n$ points $\sigma$, the minimum total cost satisfies:
    \begin{equation*}
        % \begin{array}{c}
            \OPT(\sigma) \geq \frac{n_k-1}{2} \sum_{i \in \sigma} c_i(\sigma)
        % \end{array}
    \end{equation*}
\end{lemma}
\begin{proof}
    Fix a cost-optimal clustering $\mathcal{C}^\star$ with an optimal delay profile $\mathbf{w}^\sigma = (w_i^\sigma)_{i \in \sigma}$. For any pair of distinct points $i,j \in \sigma$, $c_i(\sigma)$ and $c_j(\sigma)$ are both at most $d(\ell_i, \ell_j) + w_i^\sigma + w_j^\sigma$, and thus $d(\ell_i, \ell_j) + w_i^\sigma + w_j^\sigma \geq \frac{c_i(\sigma) + c_j(\sigma)}{2}$. As $|C| \geq n_k$ for any $C \in \mathcal{C}^\star$ since the $m$-th cluster's size satisfies $n_m \geq n_k$ for any $m \in [k]$, we have: 
    \begin{equation*}
        \begin{aligned}
            \TC(\mathcal{C}^\star, \mathbf{w}^\sigma) &= \sum_{C \in \mathcal{C}^\star} \sum_{i,j \in C: i \neq j} [d(\ell_i, \ell_j) + w_i^\sigma + w_j^\sigma] \\
            &\geq \sum_{C \in \mathcal{C}^\star} \sum_{i, j \in C : i \neq j} \frac{c_i(\sigma) + c_j(\sigma)}{2} \\
            &= \frac{1}{2} \sum_{C \in \mathcal{C}^\star} \sum_{i \in C} (|C|-1) c_i(\sigma) \geq \frac{n_k-1}{2} \sum_{i \in \sigma} c_i(\sigma)
        \end{aligned}
    \end{equation*}
\end{proof}

%We next find a lower bound on $c_i(\sigma)$ for a point $i \in \sigma^p$ so as to estimate the expected minimum total cost:

Given a sequence of points $\sigma$, we exploit the following property of a cost-optimal clustering $\mathcal{C}_{\sigma}^t$ at time $t$:
\begin{claim}
    \label{claim:delay cost of the cost-optimal clustering}
    In the UIID model, for any sequence of points $\sigma$, consider a cost-optimal clustering $\mathcal{C}_{\sigma}$ of all $n$ points with an optimal delay profile $\mathbf{w}^\sigma = (w_i^\sigma)_{i \in \sigma}$, where the intermediate clustering at time $t$ is $\mathcal{C}_{\sigma}^t$. For each point $i \in \sigma$:
    \begin{enumerate}
        \item If the cluster $\mathcal{C}_{\sigma}^{s_i}(i)$ is formed by grouping pending points, then it is created at time $s_i^\sigma = t_{j'}$ for $j' = \argmax_{j \in \mathcal{C}_{\sigma}^{s_i^\sigma}(i)} \{t_j\}$ (i.e., upon the arrival of point $j'$), while $w_i^\sigma = t_{j'}-t_i$ and $w_{j'}^\sigma = 0$.

        \item Otherwise, point $i$ is assigned to an existing cluster and $w_i^\sigma = 0$.
    \end{enumerate}
\end{claim}
\begin{proof}
    For each point $i \in \sigma$, we distinguish between two cases: 
    \begin{enumerate}
        \item If the cluster $\mathcal{C}_{\sigma}^{s_i^\sigma}(i)$ is formed at time $s_i^\sigma$ by grouping pending points, then consider the last point to arrive in $i$'s cluster, i.e., $j' = \argmax_{j \in \mathcal{C}_{\sigma}^{s_i^\sigma}(i)} \{t_j\}$. Assume, towards contradiction, that $w_{j'}^\sigma= s_{j'}^\sigma-t_{j'} > 0$, meaning that point $i$'s cluster was formed \textit{after} point $j'$ has arrived at time $s_{j'}^\sigma$. However, if point $i$'s cluster had been created at time $t_{j'}$ when point $j'$ arrived, then the total cost would have decreased. Particularly, letting $\mathbf{w}' = (w_j')_{j \in \sigma}$ be the delay profile with $w_j' = w_j^\sigma-w_{j'}^\sigma$ for each point $j \neq j'$ and $w_{j'}'=0$, then the delay profile $\mathbf{w}'$ corresponds to forming point $i$'s cluster at time $t_{j'}$ and $\TC(\mathcal{C}_{\sigma}^{n}, \mathbf{w}') > \TC(\mathcal{C}_{\sigma}^{n}, \mathbf{w}^\sigma)$, contradicting the fact that $\mathbf{w}^\sigma$ is the optimal delay profile. Accordingly, point $i$'s cluster was created at time $s_i^\sigma = t_{j'}$ (i.e., upon the arrival of point $j'$), in which case $w_i^\sigma = t_{j'}-t_i$ and $w_{j'}^\sigma = 0$. %Thus, for any pair of distinct points $i,j \in C$ with $i > j$, it holds that $w_i^\sigma + w_j^\sigma = t_{j'}-t_i + t_{j'}-t_j \geq t_i - t_j$. 

        \item Otherwise, assume that point $i$ is assigned to an existing cluster $C \in \mathcal{C}_{\sigma}^{s_i^\sigma-1}$ with $2 \leq |C|<k$ at time $s_i^\sigma$. We assume, towards contradiction, that $s_i^\sigma > t_i$. %The cluster $C$ is first formed by creating a cluster $C' \subseteq C$ of pending points at time $t := \max \{t_j: j \in C'\}$.  
        Let $C' \subseteq C$ be all points in $C$ who arrived until time $t_i-1$. If $C'$ was formed before time $t_i$ (i.e., $C' \in \mathcal{C}_{\sigma}^{t_i-1}$), then point $i$ could have been assigned to $C'$ at time $t_i$, and its delay cost would have been $w_i' = 0 < s_i^\sigma - t_i = w_i^\sigma$. Letting $\mathbf{w}' = (w_j')_{j \in \sigma}$ be the delay profile with $w_j' = w_j^\sigma$ for each point $j \neq i$ and $w_i'=0$, then $\TC(\mathcal{C}_{\sigma}^{n}, \mathbf{w}') > \TC(\mathcal{C}_{\sigma}^{n}, \mathbf{w}^\sigma)$, which is a contradiction. Hence, we infer that the points in $C'$ are still waiting during time $t_i$, and thus the cluster $C' \cup \{i\}$ could have been formed at time $t_i$ such that point $i$'s delay cost would have been $w_i' = 0 < s_i^\sigma - t_i = w_i^\sigma$, which is a contradiction by arguments similar to the above. We conclude that $s_i^\sigma = t_i$ and $w_i = 0$.
    \end{enumerate}
\end{proof}

Thus, we can now estimate the expected minimum total cost by deriving a lower bound on the minimum cost of each point:
\begin{lemma}
    \label{lemma:lower bound on each point's minimum cost}
    In the UIID model with a {\normalfont finite} metric space $\mathcal{M}=(\mathcal{X},d)$, for any sequence of $n$ locations $\sigma$ and any location $x \in \mathcal{X}$, the expected minimum cost of each point $i \in \sigma$, given that it is located at $x$, satisfies:
    \begin{equation*}
            \mathbb{E}[c_i(\sigma) | \ell_i = x] \geq \frac{1-e^{-2}}{4q_x}
    \end{equation*} 
\end{lemma}
\begin{proof}
    First, we denote by $\mathbb{E}_\sigma[\cdot]$ the expectation over the randomness of any sequence of points $\sigma$. By renaming the points, consider a random sequence of points $\sigma = (1, \dots, n)$ such that the points are ordered according to their arrival times. We then extend $\sigma$ by introducing a dummy random point $j$ for each $j \leq 0$ and $j \geq n+1$ such that point $j$ is also sampled according to $\{p_{x}\}_{x \in \mathcal{X}}$, and we thus obtain an extended random sequence $\bar{\sigma} = (\dots, -1,0,1, \dots, n, n+1, n+2, \dots)$. Hence, in the extended random sequence $\bar{\sigma}$, with probability one there are points $j \leq 0$ and $j' \geq n+1$ such that $\ell_j = \ell_{j'} = x$ for any location $x \in \mathcal{X}$. 
    
    Note that the distribution of the truncation $\bar{\sigma}_n := (1,\dots,n)$ is the same as that of $\sigma$, but the minimum cost of each point $i \in \sigma^p$ can only \textit{decrease}, i.e., $c_{i}(\bar{\sigma}) \leq c_{i}(\sigma)$, yielding $\mathbb{E}_{\bar{\sigma}}[c_{i}(\bar{\sigma}) | \ell_i = x] \leq \mathbb{E}_{\sigma}[c_{i}(\sigma)| \ell_i = x]$ for each location $x \in \mathcal{X}$. Now, observe that the (conditional) expected minimum cost of each pair of distinct points $i,j \in \bar{\sigma}$ in an extended random sequence arriving at any location $x \in \mathcal{X}$ is the same, i.e., $\mathbb{E}_{\bar{\sigma}}[c_{i}(\bar{\sigma}) | \ell_i = x] = \mathbb{E}_{\bar{\sigma}}[c_{j}(\bar{\sigma}) | \ell_j = x]$.
    
    Hence, we focus on point $0$ and analyze $\mathbb{E}_{\bar{\sigma}}[c_{0}(\bar{\sigma}) | \ell_0 = x]$ for some $x \in \mathcal{X}$. Let $\bar{\sigma}$ be an extended sequence with $\ell_0 = x$. Without loss of generality, we assume $t_0 = 0$ by shifting the points' arrival times by a constant. For the extended random sequence $\bar{\sigma}$, we denote by $\mathcal{E}^0$ the event that point $0$'s cluster is first formed during the arrival of a point $i \geq 0$ by grouping pending points. We aim to give a lower bound on point $0$'s expected minimum cost conditioned on the event $\mathcal{E}^0$. For any other point $0 \neq j \in \bar{\sigma}$, consider the case where points $0$ and $j$ are eventually in the same cluster. Under the event $\mathcal{E}^0$, by Claim \ref{claim:delay cost of the cost-optimal clustering}, point $0$'s cluster $\mathcal{C}_{\bar{\sigma}}^{s_0^{\bar{\sigma}}}(0)$ is initially formed by grouping pending points at time $s_0^{\bar{\sigma}} = t_{j'}$ for $j' = \argmax_{\tilde{j} \in \mathcal{C}_{\bar{\sigma}}^{s_0^{\bar{\sigma}}}(0)} \{t_{\tilde{j}}\}$. If point $j$ is one of those pending points (i.e., $s_j^\sigma = t_{j'}$), then:
    \begin{enumerate}
        \item If $t_{j'} \geq t_0 > t_j$: $w_0^{\bar{\sigma}} + w_j^{\bar{\sigma}} = t_{j'}-t_0 + t_{j'}-t_j \geq t_0-t_0 + t_0-t_j = t_0-t_j$ \label{case1}

        \item If $t_{j'} \geq t_j > t_0$:
        $w_0^{\bar{\sigma}} + w_j^{\bar{\sigma}} = t_{j'}-t_0 + t_{j'}-t_j \geq t_j-t_0 + t_j-t_j = t_j-t_0$ \label{case2}
    \end{enumerate}
    In both cases \eqref{case1} and \eqref{case2}, we have $w_0^{\bar{\sigma}} + w_j^{\bar{\sigma}} \geq |t_0-t_j|$. Otherwise, if point $j$ is assigned to point $i$'s cluster after its creation (i.e., $s_j^\sigma > t_{j'}$), then $w_0^{\bar{\sigma}} + w_j^{\bar{\sigma}} = t_{j'}-t_0 + s_j^\sigma-t_j > t_{j'}-t_0 + t_{j'}-t_j$, which yields that $w_0^{\bar{\sigma}} + w_j^{\bar{\sigma}} \geq |t_0-t_j|$ by arguments similar to the previous case. In any case, we have obtained that:
    \begin{equation}
        \label{eq:absolute value}
        w_0^{\bar{\sigma}} + w_j^{\bar{\sigma}} \geq |t_0-t_j| \qquad \forall 0 \neq j \in \bar{\sigma}
    \end{equation}
    
    By \eqref{eq:absolute value}, we next lower bound $\mathbb{E}_{\bar{\sigma}}[c_{0}(\bar{\sigma}) | \ell_0 = x \wedge \mathcal{E}^0]$. Let $t^-$ (resp. $t^+$) be the (finite) random waiting time between the arrival of point $0$ and the arrival of the \textit{last} point \textit{before} point $0$ (resp. the \textit{first} point \textit{after} point $0$), arriving at some location $y \in B(x,r_x)$, i.e.: 
    $$t^- := \min_{j<0} \{-t_j : d(\ell_j , x) < \rho_x\} \quad \text{ and } \quad t^+ := \min_{j>0} \{t_j : d(\ell_j , x) < \rho_x\}$$ 
    Under the event $\mathcal{E}^0$, note that, for any point $j \neq 0$, the quantity $d(\ell_i, \ell_j) + |t_j|$ can be lower bounded by $r_x$ when $d(\ell_i, \ell_j) \geq r_x$ and by $|t_j|$ otherwise. Thus, the following holds under the event $\mathcal{E}^0$:
    \begin{equation}
        \label{eq:lower bound on the minimum cost of point 0}
        \begin{aligned}
            c_0(\bar{\sigma}) &= \min_{0 \neq j \in \sigma^p} \{ d(\ell_i, \ell_j) + |t_j|\} \\
            &= \min \bigg\{ \min_{0 \neq j \in \sigma^p} \{|t_j| : d(\ell_i, \ell_j) < r_x \} , r_x \bigg\}  \\
            &= \min \bigg\{ \min_{0 > j \in \sigma^p} \{-t_j : d(\ell_i, \ell_j) < r_x \} , \min_{0 < j \in \sigma^p} \{t_j : d(\ell_i, \ell_j) < r_x \} , r_x \bigg\} \\
            &\geq \min\{\min\{t^-,t^+\}, r_x\}
        \end{aligned}
    \end{equation}

    To conclude the desired bound, note that the waiting times $t^-$ and $t^+$ are mutually independent. Since $t^-$ is the (finite) random waiting time between the arrival of point $0$ and the arrival of the \textit{last} point \textit{before} point $0$, then $t^-$ is geometrically distributed with a success probability of $q_{x} := \sum_{y \in B(x,r_x)} p_{y}$ by arguments similar to Observation 1 in the main text. That is, $t^- \sim G(q_x)$. Similarly, $t^+\sim G(q_x)$. We thereby require the following lemmas in order to derive the distribution of $\min\{t^-,t^+\}$ and the expectation of $\min\{\min\{t^-,t^+\}, r_x\}$. Combined with \eqref{eq:lower bound on the minimum cost of point 0}, this will aid us in deriving our desired lower bound on point $0$'s expected minimum cost conditioned on the event $\mathcal{E}^0$.

    \begin{lemma}
        \label{lemma:geometric RVs}
        Let $Y_1 \sim G(q_1)$ and $Y_2 \sim G(q_2)$ be two independent geometric random variables. Let $Z=\min\{Y_1,Y_2\}$. Then, $Z$ is a geometric random variable with a success probability of $1-(1-q_1)(1-q_2)$, i.e., $Z \sim G(1-(1-q_1)(1-q_2))$.
    \end{lemma}
    \begin{proof}
        Letting $s \in \mathbb{N}$, the CDFs of $Y_1$ and $Y_2$ are given by $Pr[Y_1 \leq s]=1-(1-q_1)^s$ and $Pr[Y_2 \leq s]=1-(1-q_2)^s$. Therefore, it holds that $Pr[Y_1 > s]=(1-q_1)^s$ and $Pr[Y_2 > s]=(1-q_2)^s$. Thus, the probability that both $Y_1$ and $Y_2$ are strictly greater than $s$ is $[(1-q_1)(1-q_2)]^s$, meaning that: $$Pr[Z>s]=Pr[\min\{Y_1,Y_2\}>s] = [(1-q_1)(1-q_2)]^s$$ 
        As such, the CDF of $Z$ is $Pr[Z\leq s]=Pr[\min\{Y_1,Y_2\}\leq s] =1- [(1-q_1)(1-q_2)]^s$, which is the CDF of a geometric random variable with a success probability of $1-(1-q_1)(1-q_2)$, i.e., $Z \sim G(1-(1-q_1)(1-q_2))$, as desired.
    \end{proof}

    \begin{lemma}
        \label{lemma:expected min}
        For any $s \in \mathbb{N}$ and geometric random variable $Y \sim G(q)$, it holds that $\mathbb{E}[\min\{Y,s\}]=\frac{1-(1-q)^{s}}{q}$.
    \end{lemma}
   \begin{proof}
        We simply calculate the expected value as follows:
        \begin{equation*}
            \begin{aligned}
                \mathbb{E}[\min\{Y,s\}] &= \sum_{t=1}^{s} t (1-q)^{t-1}q + \sum_{t=s+1}^\infty s (1-q)^{t-1}q \\
                &= q\sum_{t=1}^{s} t (1-q)^{t-1} + s\sum_{t=s+1}^\infty q(1-q)^{t-1} \\
                &= q \left(-\frac{\mathrm{d}}{\mathrm{d}q}\sum_{t=1}^{s} (1-q)^{t}\right) + s Pr[Y>s]\\
                &= q \left(-\frac{\mathrm{d}}{\mathrm{d}q}\frac{(1-q)[1-(1-q)^{s}]}{q}\right) + s(1-q)^s \\
                &= q \left(-\frac{(1-q)^{s}(sq+1)-1}{q^2}\right) + s(1-q)^s  \\
                &= \frac{1-(1-q)^{s}(sq+1)}{q} + s(1-q)^s \\
                &= \frac{1}{q} - s(1-q)^s - \frac{(1-q)^{s}}{q} + s(1-q)^s \\
                &= \frac{1-(1-q)^{s}}{q}
            \end{aligned}
        \end{equation*}
    \end{proof}

    As $t^- \sim G(q_x)$ and $t^+\sim G(q_x)$, then $\min\{t^-,t^+\} \sim G(1-(1-q_x)^2)$ by Lemma \ref{lemma:geometric RVs}. Together with Lemma \ref{lemma:expected min} and \eqref{eq:lower bound on the minimum cost of point 0}, we have:
    \begin{equation}
        \label{eq:expectation of Z fin2}
        \begin{aligned}
            \mathbb{E}_{\bar{\sigma}}[c_0(\bar{\sigma}) | \ell_0 = x\wedge \mathcal{E}_0]& \geq  \mathbb{E}_{\bar{\sigma}}[\min\{\min\{t^-,t^+\}, r_x\}] =\frac{1-[1-(1-(1-q_x)^2)]^{r_x}}{1-(1-q_x)^2} = \frac{1-(1-q_x)^{2r_x}}{1-(1-q_x)^2}\\
            &= \frac{1-(1-q_x)^{2r_x}}{[1-(1-q_x)][1+(1-q_x)]}=\frac{1-(1-q_x)^{2r_x}}{q_x[2-q_x]}
        \end{aligned}
    \end{equation}
    As $r_{x} \geq \frac{1}{q_x}$ by $r_{x}$'s definition in equation (2) within the main text, note that $(1-q_x)^{2r_x} \leq (1- \frac{1}{r_x})^{2r_x} \leq e^{-2}$. Combined with \eqref{eq:expectation of Z fin2} as well as $q_x[2-q_x] \leq 2q_x$ due to $q_x \in [0,1]$, we have that:
    \begin{equation}
        \label{eq:expectation of Z fin}
        \begin{aligned}
            \mathbb{E}_{\bar{\sigma}}[c_0(\bar{\sigma}) | \ell_0 = x\wedge \mathcal{E}_0] \geq\frac{1-e^{-2}}{2q_x}
        \end{aligned}
    \end{equation}
    
    Now, let $\bar{\mathcal{E}}^0$ be the complement of the event $\mathcal{E}^0$, i.e., the event that point $0$ is assigned to a cluster formed during the arrival of a point $i < 0$. By symmetry, note that $Pr[\mathcal{E}^0] = Pr[\bar{\mathcal{E}}^0] = \frac{1}{2}$ and $\mathbb{E}_{\bar{\sigma}}[c_{0}(\bar{\sigma}) | \ell_0 = x \wedge \bar{\mathcal{E}}^0] \geq 0$. Thus, we obtain the desired from \eqref{eq:expectation of Z fin}:
    \begin{equation}
        \label{eq:lower bound fin}
        \begin{aligned}
            \mathbb{E}_{\bar{\sigma}}[c_{0}(\bar{\sigma}) | \ell_0 = x] \geq Pr[\mathcal{E}^0] \mathbb{E}_{\bar{\sigma}}[c_{0}(\bar{\sigma}) | \ell_0 = x \wedge \mathcal{E}^0] \geq \frac{1-e^{-2}}{4q_x}
        \end{aligned}
    \end{equation}
\end{proof}

% Given a finite metric space $\mathcal{M}=(\mathcal{X},d)$, note that Lemma \ref{lemma:lower bound on each point's minimum cost} is affected by whether there is a location $x\in \mathcal{X}$ with $\sum_{y \in B(x,r_x)} p_{y}<1$ or not. If $\sum_{x \in \mathcal{X}} p_{x}<1$, then every location $x\in \mathcal{X}$ satisfies this property since $\sum_{y \in B(x,r_x)} p_{y}<\sum_{x \in \mathcal{X}} p_{x}<1$. However, if  $\sum_{x \in \mathcal{X}} p_{x}=1$, 

Finally, we establish a lower bound on the expected minimum total cost, which will soon allow us to establish \textsc{DGreedy}'s \textit{\textbf{constant}} ratio-of-expectations in Section \ref{sec:Ratio-of-Expectations greedy}.
\begin{theorem}
    \label{thm:lower bound on the expected minimum total cost}
    In the UIID model with a {\normalfont finite} metric space $\mathcal{M}=(\mathcal{X},d)$, the expected minimum total cost over all sequences of $n$ points sampled from a distribution $\mathcal{I}$ is at least:
    \begin{equation*}
        % \begin{array}{c}
            \opt(\mathcal{I}) \geq n(n_k-1)\frac{1-e^{-2}}{4}\sum_{x \in \mathcal{X}} \frac{p_x}{q_x} 
        % \end{array}
    \end{equation*}
    where $q_x:=1-\sum_{y \in B(x,r_x)} p_{y}> 1-\sum_{x \in \mathcal{X}} p_{x}> 0$.
\end{theorem}
\begin{proof}
    By Lemmas \ref{lemma:the optimal cost is at least half the overall minimum cost}-\ref{lemma:lower bound on each point's minimum cost}, $\opt(\mathcal{I})=\mathbb{E}[\OPT(\sigma)]$ satisfies: 
    \begin{equation*}
        \begin{aligned}
            \mathbb{E}[\OPT(\sigma)] \geq \frac{n_k-1}{2} \sum_{i \in N} \mathbb{E}[c_{i}(\sigma)] = \frac{n_k-1}{2} \sum_{i \in N} \sum_{x \in \mathcal{X}} \mathbb{P}[\ell_{i} = x] \mathbb{E}[c_{i}(\sigma) | \ell_{i} = x] \geq n(n_k-1)\frac{1-e^{-2}}{4}\sum_{x \in \mathcal{X}} \frac{p_x}{q_x}
        \end{aligned}
    \end{equation*}
\end{proof}

\subsection{\textsc{DGreedy}'s Constant Ratio-of-Expectations}
\label{sec:Ratio-of-Expectations greedy}
We herein prove that the performance guarantee within the UIID model is significantly better compared with the current best competitiveness in the (worst-case) adversarial model. Particularly, we show that \textsc{DGreedy} achieves a \textbf{\textit{constant}} ratio-of-expectations:
\begin{theorem}
    \label{thm:ratio-of-expectations of dgreedy}
    In the UIID model with a {\normalfont finite} metric space $\mathcal{M}=(\mathcal{X},d)$, if the cluster sizes $\{n_m\}_{m\in [k]}$ are \textit{constants} and the number of clusters $k$ grows with the number of points $n$ ($n,k$ are not constants), \textsc{DGreedy} (Algorithm \ref{alg:delayed greedy}) has a {\normalfont \textbf{constant}} ratio-of-expectations of:
    \begin{equation*}
        % \begin{array}{c}
            \lim_{n \rightarrow \infty} \frac{\text{\textsc{DGreedy}}(\mathcal{I})}{\opt(\mathcal{I})} \leq \frac{8 (n_1-1)}{(n_k-1)(1-e^{-2})} 
        % \end{array}
    \end{equation*}
    If all cluster sizes are equal (i.e., $n_m=n_{m'}$ for any $m,m'\in[k]$), then:
    \begin{equation*}
        % \begin{array}{c}
            \lim_{n \rightarrow \infty} \frac{\text{\textsc{DGreedy}}(\mathcal{I})}{\opt(\mathcal{I})} \leq \frac{8 }{1-e^{-2}} 
        % \end{array}
    \end{equation*}
\end{theorem}
\begin{proof}
    By Theorem \ref{thm:upper bound on the greedy expected total cost}:
    \begin{equation}
        \label{eq:roe finite}
        % \begin{array}{cl}
            \lim_{n \rightarrow \infty} \frac{\text{\textsc{DGreedy}}(\mathcal{I})}{\opt(\mathcal{I})} \leq \lim_{n \rightarrow \infty} \frac{2n (n_1-1) \sum_{x \in \mathcal{X}} p_x r_{x}}{\opt(\mathcal{I})}  + \lim_{n \rightarrow \infty}\frac{2(n_1-1) |\mathcal{X}|\cdot\left[d_{\max}+ \frac{1}{\sum_{x \in \mathcal{X}} p_{x}}\right]}{\opt(\mathcal{I})}
        % \end{array}
    \end{equation}
    By Theorem \ref{thm:lower bound on the expected minimum total cost}, the expected minimal total cost $\opt(\mathcal{I})$ is lower bounded by a quantity linear in the number of points $n$. Hence, %regardless of whether $\sum_{x \in \mathcal{X}} p_{x}=1$ or not, 
    since the numerator of the last term in \eqref{eq:roe finite} does not depend on $n$, while $n_1$ is constant and we consider finite metric spaces (i.e., $|\mathcal{X}|<\infty$), the limit in the last term in \eqref{eq:roe finite} equals $0$. By also applying Theorem \ref{thm:lower bound on the expected minimum total cost} to the first term on the right-hand-side of \eqref{eq:roe finite}:
    \begin{equation}
        \label{eq:almost final}
        \begin{aligned}
            \lim_{n \rightarrow \infty} \frac{\text{\textsc{DGreedy}}(\mathcal{I})}{\opt(\mathcal{I})} &\leq \lim_{n \rightarrow \infty} \frac{2n (n_1-1) \sum_{x \in \mathcal{X}} p_x r_{x}}{ n(n_k-1)\frac{1-e^{-2}}{4}\sum_{x \in \mathcal{X}} \frac{p_x}{q_x}} =\frac{2 (n_1-1)\sum_{x \in \mathcal{X}} p_x r_{x}}{ (n_k-1)\frac{1-e^{-2}}{4}\sum_{x \in \mathcal{X}} \frac{p_x}{q_x}} \\
            &\leq \frac{8 (n_1-1)\sum_{x \in \mathcal{X}} \frac{p_x}{q_x}}{(n_k-1)(1-e^{-2})\sum_{x \in \mathcal{X}} \frac{p_x}{q_x}} = \frac{8 (n_1-1)}{(n_k-1)(1-e^{-2})} 
        \end{aligned}
    \end{equation}
    where the last inequality follows from $\frac{1}{\sum_{y \in \bar{B}(x,r)} p_{y}}\leq r_{x} \leq \frac{1}{q_x} = \frac{1}{\sum_{y \in {B}(x,r)} p_{y}}$ due to $r_x$'s definition in \eqref{eq:radius at locations x}. If all cluster sizes are equal (i.e., $n_m=n_{m'}$ for any $m,m'\in[k]$), then the above yields:
    \begin{equation*}
        \lim_{n \rightarrow \infty} \frac{\text{\textsc{DGreedy}}(\mathcal{I})}{\opt(\mathcal{I})} \leq \frac{8}{1-e^{-2}}
    \end{equation*}
\end{proof}

\begin{remark}[Finite $n$]
    As noted, e.g., after \eqref{eq:roe}, the RoE measures an algorithm’s asymptotic behavior as $n\rightarrow\infty$ captures its efficiency in practical large-scale scenarios instead of small pathological worst-case situations that may not truly reflect its effectiveness \cite{mari2024online}. However, in practical online settings, systems operate with \textit{finite}, often bounded, inputs. One can derive concrete finite-$n$ bounds directly from our analysis if preferred: Theorems \ref{thm:lower bound on the expected minimum total cost} and \ref{thm:ratio-of-expectations of dgreedy} already give finite-$n$ bounds, while DGreedy’s RoE for finite $n$ easily follows from \eqref{eq:roe finite}-\eqref{eq:almost final}.
\end{remark}

\section{Extensions}
In this section, we show that our theoretical guarantees easily extend to non-metric spaces that encompass certain subclasses of hedonic games (Section \ref{sec:Extension to Non-Metric Spaces}), and cases with lower and upper bounds on cluster sizes (Section \ref{sec:Lower and Upper Bounds on Cluster Sizes}).

\subsection{Non-Metric Spaces and Hedonic Games}
\label{sec:Extension to Non-Metric Spaces}
To the best of our knowledge, online non-center clustering with delays has not been studied before, let alone in (non-)metric spaces \cite{cohen2024online,cohen2024onlinefriends,cohen2023online,cohen2025online,cohen2025decentralized,cohen2025fair,flammini2021online,bullinger2025online}. Yet, while our algorithm’s analysis relies on the metric assumption, it employs several non-trivial components that also extend to non-metric domains. Indeed, consider a \textit{non-metric} space $\mathcal{M}=(\mathcal{X},d)$. First, assume that $d(x,y)>0$ for any $x\neq y$.
\begin{enumerate}
    \item If $d$ satisfies the triangle inequality but is asymmetric (i.e., $d(x,y)\neq d(y,x)$ for some $x,y$), then we can run our algorithm in the symmetrized metric space $\mathcal{M}^S=(\mathcal{X},d^S)$ where $d^S(x,y)=\frac{d(x,y)+d(y,x)}{2}$ for any $x,y$, replacing each occurrence of $d$ in the connection cost, total cost, the algorithm, and its analysis with $d^S$. As $\mathcal{M}^S$ is a metric space, all guarantees continue to hold. \label{item:symmetrization}
    
    \item If $d$ does not satisfy the triangle inequality but is still symmetric (i.e., $d(x,y)=d(y,x)$ for any $x,y$), then construct the weighted complete directed graph $G=(\mathcal{X},\mathcal{X}\times\mathcal{X},d)$ with weight $d(x,y)$ on arc $(x,y)$ and let $d_G(x,y)$ be the length of the shortest path between $x,y$. As $d$ is symmetric, $d_G$ is a metric. Thus, running our algorithm in the metric space $\mathcal{M}_G=(\mathcal{X},d_G)$ maintains our guarantees as before.

    \item If $d$ is asymmetric and does not satisfy the triangle inequality, we consider the symmetrization from case \eqref{item:symmetrization} of $\mathcal{M}_G$. \label{item:asym and not triangle}

    \item Now, if $d(x,y)<0$ for some $x\neq y$, replace $d$ with $d’(x,y)=d(x,y)-\min_{x,y}d(x,y)+\epsilon$ for some tiny $\epsilon>0$ to ensure $d’(x,y)>0$, shifting to nonnegatives while preserving ordering of similarity. Then, apply case \eqref{item:asym and not triangle} above.
\end{enumerate}

\subsubsection{Online Additively Separable Hedonic Games with Agent Types}
\label{sec:Online Additively Separable Hedonic Games with Agent Types}
Both a metric space and a non-metric space can model online additively separable hedonic games (ASHGs) with agent types (see \cite{cohen2024online,cohen2024onlinefriends,cohen2023online,cohen2025online,cohen2025decentralized,cohen2025fair,flammini2021online,bullinger2025online} for works on online additively separable hedonic games). Below, we outline the minor adjustments in terminology and modeling required to capture this specific application. After those slight modifications, our results readily extend to online ASHGs with agent types, as explained earlier.

In \textit{additively separable hedonic games} (ASHGs), an agent assigns a numerical valuation to each other agent, indicating the intensity by which she prefers one agent to another. Formally, an online ASHG with agent types is characterized by a (non)-metric space $\mathcal{M}=(\mathcal{X},d)$. Here, each location $x \in \mathcal{X}$ can be seen as a possible \textit{agent type}, where the preferences of an agent of type $x$ are encoded by a cardinal disutility function $d(x,\cdot)$ (i.e., the negation of her utility), specifying that an agent of type $x$ assigns a cardinal disutility of $d(x,y)$ to each agent of type $y$. The input to an online ASHG with agent types is then a sequence $\sigma$ of $n$ agents that arrive online over multiple rounds. Each agent $i \in \sigma$ is given by her \textit{type} $\ell_i \in \mathcal{X}$ and \textit{arrival time} $t_i \in \mathbb{N}^+$. As before, to generate the sequence $\sigma$ of $n$ agents, we can also consider the \textit{unknown i.i.d.} (UIID) model, where arrivals are sampled independently across time from a \textit{unknown}, \textit{fixed} probability distribution $\{p_{x}\}_{x \in \mathcal{X}}$ satisfying $\sum_{x \in \mathcal{X}} p_{x} \leq 1$. In particular, the probability that each agent $i$ arrives with type $x \in \mathcal{X}$ is $\mathbb{P}[\ell_i=x] = p_{x}$. However, no agent arrives between any two consecutive arrivals of agents with probability $1-\sum_{x \in \mathcal{X}} p_{x}$, which is positive if $\sum_{x \in \mathcal{X}} p_{x} < 1$. We denote by $N^t$ the set of agents that arrive until time $t \in [T]$, and we set $N := N^T = N^t$ for any time $t \geq T$.

As in Section \ref{sec:Preliminaries}, at each time $t \in [T]$, an online algorithm $\mathcal{A}$ shall produce a \textit{partial partition} $\mathcal{C}^t$ of the agents \textit{who arrived until time $t$} into disjoint subsets of agents (i.e., \textit{coalitions}), without any knowledge about future agents. Once agent $i$ arrives at time $t_i$, her location $\ell_i$ is revealed. Then, $\mathcal{A}$ has to \textit{irrevocably} assign agent $i$ to an existing coalition in $\mathcal{C}^{t-1}$ or create a new coalition $\{t\}$, if possible. Unlike prior research on online clustering and online coalition formation \cite{cohen2024online,cohen2024onlinefriends,cohen2023online,cohen2025online,cohen2025decentralized,cohen2025fair,flammini2021online,bullinger2025online,cohen2022online}, in our work an online algorithm $\mathcal{A}$ does \textit{not} have to \textit{immediately} assign an arriving agent $i$ to a coalition, but it may postpone the decision to some time $s_i \geq t_i$ with a \textit{\textbf{delay cost}}, which equals agent $i$'s waiting time, denoted as $w_i := s_i-t_i$. However, if no assignment is made at time $t\in [T]$, then $\mathcal{C}^{t}=\mathcal{C}^{t-1}$. Moreover, if agents remain unassigned once the last agent arrives (i.e., at time $t_n$), then they will keep waiting until the algorithm inserts them into coalitions. We thus denote the number of coalitions in $\mathcal{C}^t$ as $|\mathcal{C}^t|$ and the coalition in $\mathcal{C}^t$ containing agent $i \in N^t$ as $\mathcal{C}^t(i)$.  

Once two distinct agents $i \neq j$ are in the same coalition, agent $i$ suffers a \textit{\textbf{connection cost}}, given by her disutility from agent $j$ plus their delay costs, i.e., $d(\ell_i, \ell_j) + w_i + w_j$. Particularly, as customary in ASHGs, to obtain preferences over coalitions, the disutility of each agent $i$ from her coalition $\mathcal{C}^t(i)$ is aggregated by summing her connection costs from other coalition members, given by $\sum_{j\in \mathcal{C}^t(i)\setminus \{i\}} [d(\ell_i, \ell_j) + w_i + w_j]$. Therefore, we evaluate the quality of the partition $\mathcal{C}^t$ at time $t\in [T]$ by its \textit{\textbf{total cost}}, defined as the sum of the connection costs in each coalition, i.e., $\TC(\mathcal{C}^t, \mathbf{w}) = \sum_{C \in \mathcal{C}^t} \sum_{i,j \in C: i \neq j} [d(\ell_i, \ell_j) + w_i + w_j]$, where $\mathbf{w} = (w_i)_{i \in N}$ is the \textit{delay profile}. For any coalition $C \in \mathcal{C}^t$ and agent $i\in C$, agent $i$'s delay cost affects her connection cost to any other agent $j \in C$. The main goal of $\mathcal{A}$ is thus finding a \textit{cost-optimal} clustering $\mathcal{C}^\star$ that \textit{minimizes} the total cost among all possible partitions of the $n$ agents. All other preliminaries described in Section \ref{sec:Preliminaries} for online non-centroid clustering are identical.

\subsection{Lower and Upper Bounds on Cluster Sizes}
\label{sec:Lower and Upper Bounds on Cluster Sizes}
Our results can be extended to cases where each cluster $m$ may have any size in an interval $[l_m,u_m]$, provided that $\sum_m l_m \leq n \leq \sum_m u_m$ for feasibility and $l_{\min}:=\min_m l_m\geq 2$ with $l_m,u_m$ treated as constant integers. To enforce the lower bounds, we first run DGreedy treating $l_m$ as the size of the cluster $m$, yielding a clustering $\mathcal{C}=(C_m)$ with $|C_m|=l_m$. If $n=\sum_m l_m$, we stop; otherwise, we rerun DGreedy treating $u_m$ as the size of the cluster $m$ with $\mathcal{C}$ as the initial clustering; the algorithm already enforces the upper bounds. Our proofs then easily extend while replacing $n_1$ with $u_{\max}:=\max_m u_m$ and $n_k$ with $l_{\min}$ in Lemmas \ref{lemma:at most twice the delay cost} and \ref{lemma:the optimal cost is at least half the overall minimum cost}, respectively, including Theorems \ref{thm:lower bound on the expected minimum total cost} and \ref{thm:ratio-of-expectations of dgreedy} that are derived from them.

\section{Conclusions and Future Work}
We presented a new model for studying online non-centroid clustering \textit{with delays}, where elements that arrive sequentially as points in a finite metric space must be assigned to clusters, but decisions may be delayed at a cost. While the classic \textit{worst-case model} is too pessimistic even in restricted  cases, we studied the \textit{unknown i.i.d. model}, under which we developed an online algorithm with a \textit{\textbf{constant}} ratio-of-expectations, offering hope for beyond worst-case analysis and the design of effective algorithms in practice. 

Our research paves the way for many future works. Immediate directions are studying general delay costs as well as cases where the distribution of the arriving points' locations is known and/or may change over time. Finally, it is worth examining settings where clusterings can be modified with a penalty.

%%%%%%%%%%%%%%%%%%%%%%%%%%%%%%%%%%%%%%%%%%%%%%%%%%%%%%%%%%%%%%%%%%%%%%%%

%%% The acknowledgments section is defined using the "acks" environment
%%% (rather than an unnumbered section). The use of this environment 
%%% ensures the proper identification of the section in the article 
%%% metadata as well as the consistent spelling of the heading.

\section*{Acknowledgments}
The author acknowledges travel support from a Schmidt Sciences 2025 Senior Fellows award to Michael Wooldridge.

\bibliographystyle{plainnat} 
\bibliography{example_paper}

\end{document}